\documentclass[12pt,reqno]{article}

\usepackage[usenames]{color}
\usepackage{amssymb}
\usepackage{graphicx}
\usepackage{amscd}

\usepackage[colorlinks=true,
linkcolor=webgreen,
filecolor=webbrown,
citecolor=webgreen]{hyperref}

\definecolor{webgreen}{rgb}{0,.5,0}
\definecolor{webbrown}{rgb}{.6,0,0}
\usepackage{color}
\usepackage{fullpage}
\usepackage{float}

\usepackage{graphics,amsmath,amssymb}
\usepackage{amsthm}
\usepackage{amsfonts}
\usepackage{latexsym}
\usepackage{epsf}
\usepackage{amscd}
\usepackage{fullpage}
\usepackage{float}
\usepackage{url}

\title{Lengths of Words Accepted by Nondeterministic Finite Automata}

\author{
Aaron Potechin\footnote{This work was supported by the Simons Collaboration for Algorithms and Geometry and by the NSF under agreement No.~CCF-1412958.}\\
KTH Royal Institute of Technology\\
Stockholm 114-28\\
Sweden\\
\href{mailto:aaronpotechin@gmail.com}{\tt aaronpotechin@gmail.com}   \\
\ \\
Jeffrey Shallit\footnote{Supported in part by NSERC Grant 105829/2013.} \\
School of Computer Science \\
University of Waterloo \\
Waterloo, ON  N2L 3G1 \\
Canada \\
\href{mailto:shallit@uwaterloo.ca}{\tt shallit@uwaterloo.ca}
}

\begin{document}

\theoremstyle{plain}
\newtheorem{theorem}{Theorem}
\newtheorem{corollary}[theorem]{Corollary}
\newtheorem{lemma}[theorem]{Lemma}
\newtheorem{proposition}[theorem]{Proposition}

\theoremstyle{definition}
\newtheorem{definition}[theorem]{Definition}
\newtheorem{example}[theorem]{Example}
\newtheorem{conjecture}[theorem]{Conjecture}

\theoremstyle{remark}
\newtheorem{remark}[theorem]{Remark}

\maketitle

\begin{abstract}
We consider two natural problems about nondeterministic finite automata.
First, given such an automaton $M$ of $n$ states, and a length $\ell$, does
$M$ accept a word of length $\ell$?  
We show that the classic problem of triangle-free graph recognition
reduces to this problem, and give an 
$O(n^\omega (\log n)^{1+\epsilon} \log \ell)$-time
algorithm to solve it, where $\omega$ is the optimal exponent
for matrix multiplication.  Second, provided $L(M)$ is finite, we consider
the problem of listing the lengths of {\it all\/} words accepted by $M$.
Although this problem seems like it might be significantly
harder, we show that this problem
can be solved in $O(n^\omega (\log n)^{2+\epsilon})$ time.
Finally, we give a connection between NFA acceptance and the strong
exponential-time hypothesis.
\end{abstract}


\section{Introduction}
\label{one}

A nondeterministic finite automaton (NFA) $A = (Q, \Sigma, \delta, q_0, F)$
consists of a finite, nonempty set of states $Q = \{ q_0, q_1, \ldots,
q_{n-1} \}$, an input alphabet
$\Sigma$, an initial state $q_0$, a set $F \subseteq Q$ of final states,
and a transition function $\delta: Q \times \Sigma \rightarrow 2^Q$.
This transition function is then extended in the usual way,
to the domain $Q \times \Sigma^*$.  The language accepted by an NFA $A$
is defined to be
$$ L(A) = \{ x \in \Sigma^* \ : \ \delta(q_0,x) \, \cap \, F \not= \emptyset \} .$$
Our NFA's do not have
$\epsilon$-transitions.  The {\it transition diagram} of an NFA $A$
is the directed graph $G = G(A)$ with source $q_0$,
sink vertices given by $F$,
and directed edge from $p$ to $q$ labeled $a$ if $\delta(p,a) = q$.
An NFA is {\it unary} if its input alphabet $\Sigma$ consists of a 
single letter.  
For more information about the model, 
the reader can consult, for example, \cite{Hopcroft&Ullman:1979}.

Without loss of generality, we can assume all NFA's under discussion are
{\it initially connected} (i.e., every state is reachable from the start
state $q_0$) and that a final state is reachable from every state.
Note that both properties are testable for an NFA $A$ in time linear in the
number of edges in its transition diagram.

An NFA $A$ is {\it acyclic} if its transition diagram
has no cycles, or, equivalently, if $L(A)$ is finite.
Note that if an $n$-state NFA is acyclic, then $L(A) \subseteq (\Sigma \, \cup \, \{\epsilon \})^{n-1}$.

In this note, we consider the following natural problem about NFA's:

\bigskip

\noindent{\tt NFA LENGTH ACCEPTANCE}

\medskip

\noindent {\it Instance:}  An $n$-state NFA $A = (Q, \Sigma, \delta, q_0, F)$ 
and a length $\ell$.

\medskip

\noindent {\it Question:}  Does $A$ accept a word of length $\ell$?

\bigskip

\begin{proposition}
{\tt NFA LENGTH ACCEPTANCE} can be solved in $O(n^\omega (\log n)^{1+\epsilon}
(\log \ell))$ time. 
\end{proposition}

\begin{proof}
First, we create a boolean adjacency 
matrix $M = M(A)$ with a $1$ in
row $i$ and column $j$ if there is a letter $a$ such that
$q_j \in \delta(q_i, a)$.  Then standard results on path algebras
imply that $A$ accepts a word of length $\ell$ if and only if
$(M^\ell)_{0,j} = 1$ for some $j$ such that $q_j \in F$.  
A single Boolean matrix multiplication can be carried out using
the usual matrix multiplication algorithms modulo $n+1$, and then 
converting each element that is $\geq 1$ to $1$.  This involves
arithmetic on integers of $\log n$ bits (which can be done in
$(\log n)^{1 + \epsilon}$ time).  Raising $M$ to the $\ell$ power can
be done using the usual ``binary method of exponentiation'' 
(see, e.g., \cite[\S 4.6.3]{Knuth:1981}) with
$\log \ell$ matrix multiplications.  
\end{proof}

In the next section, we prove a lower bound on the
complexity of this problem, by reducing from the classic problem of
triangle-free graph recognition.  The same reduction works even if
our NFA is restricted to be over a unary alphabet, and even if
it is required to be acyclic.  In Section~\ref{alg} we discuss the
problem of listing all the elements of $L(A)$ when $A$ is a unary
acyclic NFA.

\section{A lower bound}
\label{lower}

In this section, we show that 
the classic problem of determining whether
an undirected graph is triangle-free reduces to  
{\tt NFA LENGTH ACCEPTANCE} in linear time.  

Let $G$ be an undirected graph on $n$ vertices,
say $v_0, v_1, \ldots, v_{n-1}$.
We assume, without loss of generality, that $G$ has no self-loops.
We create a
unary acyclic NFA $A$ as follows.  The construction consists of
four layers, numbered from 1 to $4$, with each layer having $n$ states, each
corresponding to one of $G$'s vertices.
State $i$ in layer $j$ is denoted $q_i^j$.
In the top layer (layer 1), we let $q_0^1$ be the initial state of $A$ and we
add a transition from $q_i^1$ to $q_{i+1}^1$ for $0 \leq i \leq n-2$, giving
us a linearly-connected chain of states.
Next, we add transitions from layer 1 to layer 2, layer 2 to layer 3, and
layer 3 to layer 4 as follows:  if $G$ has an edge from $v_i$ to $v_k$,
then $A$ has a transition from $q_i^j$ to $q_k^{j+1}$ for $j = 1,2,3$.
Finally, the bottom layer (layer 4) has transitions
from $q_i^4$ to $q_{i+1}^4$ for $0 \leq i \leq n-2$.  
The unique final state of $A$ is $q_{n-1}^4$.    A similar idea
was used in \cite{Abboud&Williams:2014}.

We claim that $a^{n+2}$ is accepted by $A$ if and only if there exists a
triangle in $G$.

\begin{figure}[ht]
\centerline{\includegraphics[height=10cm]{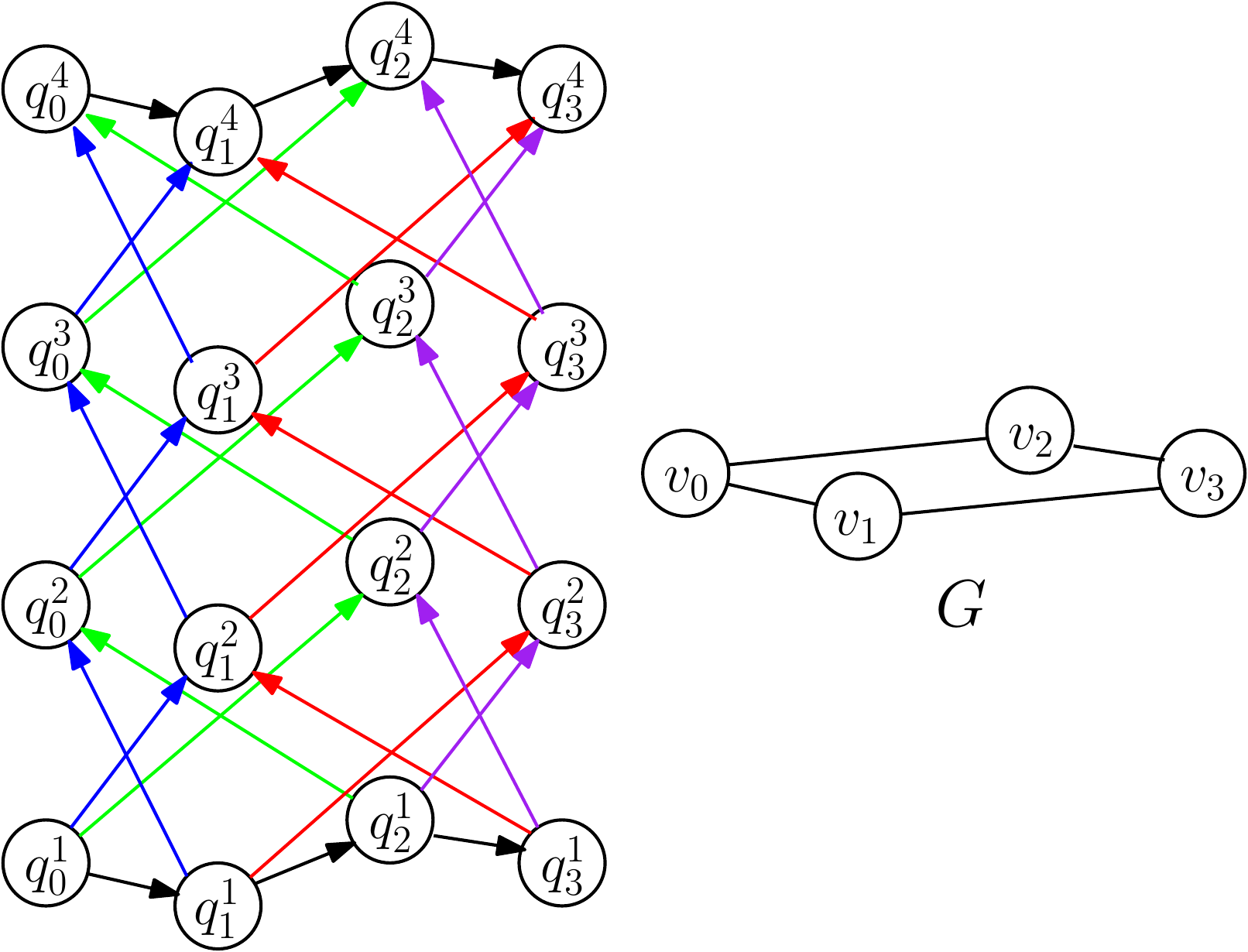}}
\caption{The reduction when $G$ is a square. Since a square has no triangle, there is no path of length $6$ from $q^1_0$ to $q^4_3$.}
\label{trianglereduction}
\end{figure}

Suppose a word $a^r$ is accepted by $A$.  Then an accepting path must begin
at the initial state $q_0^1$, follow $i$ edges in layer $1$, ending
at $q_i^1$
then transit to layer 2, arriving at
state $q_j^2$, then transit to layer 3, arriving at state
$q_k^3$, then transit to layer 4, arriving at state $q_l^4$
and finally, end at $q_{n-1}^4$.  The length of this path
is then $r = i+3+n-1-l$.  But $r = n+2$ if and only
if $i = l$.  Then $G$ has edges $(v_i, v_j)$, $(v_j, v_k)$, and
$(v_k, v_i)$ and so has the triangle $(v_i, v_j, v_k)$.

Now suppose $G$ has the triangle  $(v_i, v_j, v_k)$.  Then there
are edges $(v_i, v_j)$, $(v_j, v_k)$, and $(v_k, v_i)$.  An acceptance
path for $a^{n+2}$ is as follows:  from
$q_0^1$ in a linear chain of nodes to $q_i^1$ by a path of length $i$, a transition
to $q_j^2$, a transition to $q_k^3$, a transition to
$q_i^4$, and transitions to $q_{n-1}^4$ in a linear chain of nodes by a path
of length $n-1-i$.  The accepted word has length $i + 3 + n-1-i = n+2$.

Starting with a graph $G$ of $n$ vertices and $m$ edges, this construction
produces a unary acyclic NFA with
$4n$ vertices and $3m+2n$ edges.  

We have proved:

\begin{theorem}
There is a linear-time reduction from {\tt TRIANGLE-FREE GRAPH} to 
{\tt NFA LENGTH ACCEPTANCE}.
\end{theorem}

The fastest general algorithm for {\tt TRIANGLE-FREE GRAPH} known runs
in $O(n^\omega (\log n)^{1+\epsilon})$ time \cite{Itai&Rodeh:1978,Alon&Yuster&Zwick:1995,Alon&Yuster&Zwick:1997}.
It consists of computing
$M^3$, where $M$ is $G$'s adjacency matrix,
and checking if the diagonal contains a $1$.
This suggests that finding a significantly faster
algorithm for {\tt NFA LENGTH ACCEPTANCE} will require a large advance.

\section{Unary acyclic NFA enumeration}
\label{alg}

In this section we consider a related problem, which we call
{\tt UNARY ACYCLIC NFA ENUMERATION}:

\medskip

\noindent{\it Instance:}  a unary acyclic $n$-state NFA $A$.

\medskip

\noindent{\it Problem:}  to enumerate (list) the elements of $L(A)$. 

\bigskip

At first glance, this problem seems like it might be harder than
{\tt NFA LENGTH ACCEPTANCE}, since it requires checking the lengths
of all possible accepted words, rather than a single word.
Nevertheless, we give
a $O(n^\omega (\log n)^{2+ \epsilon})$-time 
algorithm for the problem.   Since
the same argument giving a linear-time reduction
from {\tt TRIANGLE-FREE GRAPH} to {\tt NFA LENGTH ACCEPTANCE}
works for reducing {\tt TRIANGLE-FREE GRAPH} to
{\tt UNARY ACYCLIC NFA ENUMERATION}, it is unlikely we can greatly improve
our algorithm, unless a significant advance is made.

The naive approach to solving 
{\tt UNARY ACYCLIC NFA ENUMERATION}
is to maintain a list $L$ of the states of $A$ (represented, say, as a bit vector)
and update this list as we read additional
symbols of input. If $A$ has $n$ states, then the longest word accepted
is of length $\leq n-1$.   To update $L$ after reading each new symbol
potentially requires a union of $n$ sets, each with at most $n$ elements.
Thus the total running time is $O(n^3)$.

We consider a different approach.
Suppose $A$ has $n$ states, labeled $q_0, q_1, \ldots, q_{n-1}$.
We create a new NFA $A' = (Q', \{ a \}, \delta', q'_0, F')$, as follows.
Let $2^k$ be the smallest power of $2$ that is $\geq n$.
Define $Q' = Q \, \cup \, \{ p_0, p_1, \ldots, p_{2^k-1} \}$.
Let $q'_0 = p_0$ be the new initial state, and, in addition to the transitions
already present in $A$, define $\delta'$ by adding additional transitions from
$p_i$ to $p_{i+1}$ for $0 \leq i < p_{2^k -1}$, and from
$p_{2^k -1}$ to $q_0$.    Let $M'$ be the adjacency matrix of $A'$.

Now $A$ accepts
a word of length $i$ if and only if there is a path of length $i$
from $q_0$ to a final state of $A$, if and only if there is a path
of length $2^k$ from $p_i$ to a final state of $A'$.    Thus we can
compute all words accepted by $A$ with a single exponentiation of
$M'$ to the appropriate power.

We now compute ${M'}^{2^k}$ using exactly $k$ Boolean matrix multiplications,
through repeated squaring.  To determine if $a^i$ is accepted by $A$,
it suffices to check the entry
corresponding to the row for $p_i$ and the columns for the final states of
$A'$.  We do this for each possible length, $0$ through $n-1$, and
so the total cost is $O(n^\omega (\log n)^{2+\epsilon} + n^2)$.

We have proved

\begin{theorem}
     If $M$ is a unary NFA that accepts a finite language $L$, we can 
enumerate the elements of $L$ in $O(n^\omega (\log n)^{2+\epsilon})$
bit operations, where $\omega$ is the optimal exponent for matrix
multiplication.
\end{theorem}

This result previously appeared in \cite[\S 3.8]{Shallit:2009}.

\section{Hardness of NFA acceptance}

In this section, we consider the following decision problem:

\medskip

\noindent{\tt NFA ACCEPTANCE}

\medskip

\noindent{\it Input:}  An NFA $M$ of total size $m$ (states and transitions)
and an input $x$ of length $\ell$.

\medskip

\noindent{\it Question:}  Does $M$ accept $x$?

\medskip

The obvious algorithm for this problem keeps track of the current set of
states and updates it for each new input letter read; it runs
in $O(\ell m)$ time.

In this section, we show that in the case when the NFA is sparse (i.e.,
$m$ is not much larger than $n$, the number of states of the NFA),
significantly improving this algorithm would disprove the strong
exponential time hypothesis (SETH) \cite{Impagliazzo&Paturi:1999}.
However, this does not rule out an improvement when the
NFA is dense, and we leave it as an open problem to either find a
significant improvement to this algorithm, or show why such an
improvement is unlikely.

Recall the following decision problem (e.g., \cite{Williams&Yu:2014}):

\medskip

\noindent {\tt ORTHOGONAL VECTORS} 

\medskip

\noindent Input:  Two lists $(v_i)_{1 \leq i \leq n}$ and
$(w_i)_{1 \leq i \leq n}$ of boolean vectors of dimension $d$.

\medskip

\noindent Question:  do there exist $i, j$ such that
the boolean product $v_i \cdot w_j = 0$?

\bigskip

\begin{theorem}
{\tt ORTHOGONAL VECTORS} reduces in linear time and log space to acyclic
{\tt NFA ACCEPTANCE}.
\end{theorem}

\begin{proof}
The idea is to create an NFA $M$ that accepts the input $00w_100w_2\cdots 00w_n$
if and only if there exist $i,j$ such that $v_i \cdot w_j = 0$.

The NFA is built out of some simple DFA gadgets $M_i$, one for each $v_i$.
On input $w$, the DFA $M_i$ accepts iff $w \cdot v_i = 0$.  If the vectors are of length $d$, this can
be done with $2d+1$ states.  

\begin{figure}[ht]
\centerline{\includegraphics[width=10cm]{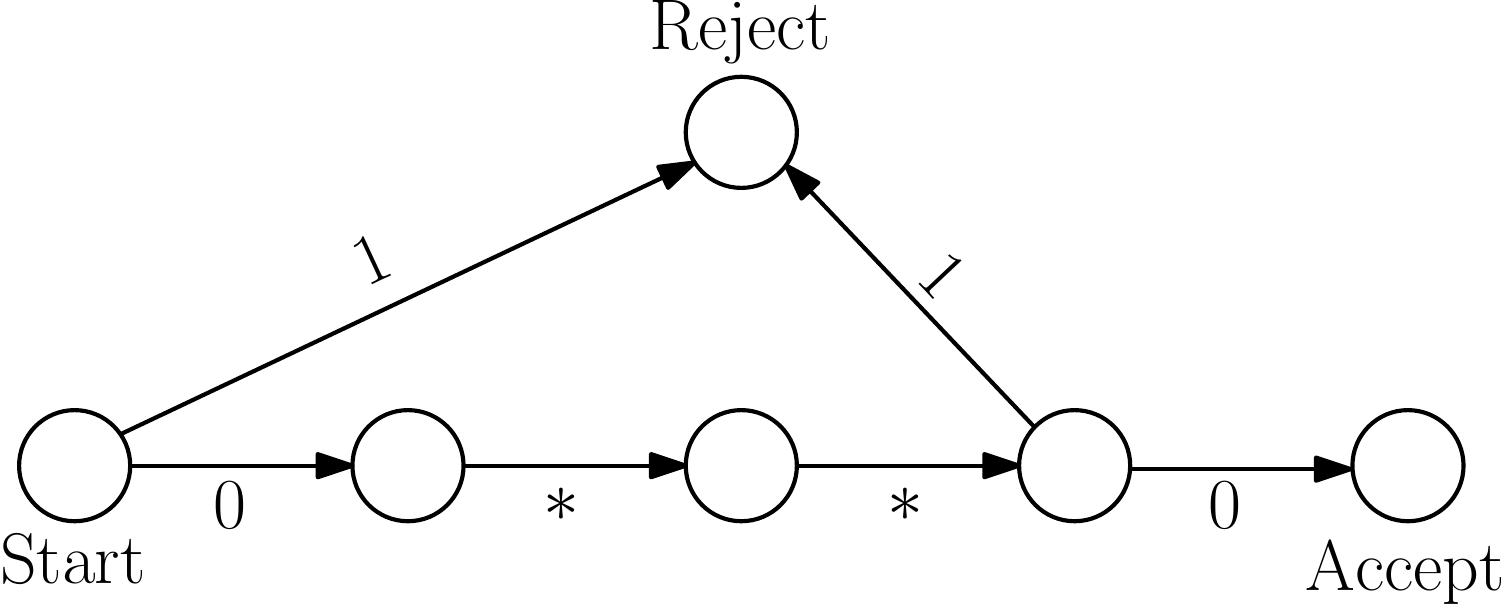}}
\caption{The gadget for testing whether $w \cdot v$ is $0$ when $w \in \{0,1\}^4$ and $v = \{1,0,0,1\}$. Here $*$ means that the NFA can take this edge regardless of what the input bit is.}
\label{orthogonalitygadget}
\end{figure}

The NFA $M$ has the following layers:
\begin{enumerate}
\item A path of length $(n-1)(d+2)$, where we assign the label $a_{j}$ to the $(j-1)(d+2)$-th state on this path. We set the start state to be $a_1$.
\item A special state $x$.
\item The gadgets $M_i$, except for their accept states.
\item A special state $y$ that replaces the accepting state for each gadget. 
\item A path of length $(n-1)(d+2)$, where we assign the label $b_{j}$ to the $(j-1)(d+2)$-th state on this path. We set the accept state to be $b_n$.
\end{enumerate} 
The transitions for $M$ are as follows. Except for transitions within the gadgets $M_j$, all of these transitions can be made regardless of the input.
\begin{enumerate}
\item For the path containing the states $a_j$, at each $a_j$ we either choose to transition to $x$, which means that we read $w_j$ from the input, or we can transition to the next state of the path (unless we are at $a_n$, in which case we can only transition to $x$). At each other state in the path, we can only transition to the next state of the path.
\item From $x$ we can transition to the start state of any gadget $M_i$, which means that we will check $v_i \cdot w_j$.
\item We have the transitions for each gadget $M_i$, except that if we would transition to the accept state of $M_i$, we instead transition to $y$.
\item From $y$ we can transition to any $b_j$. This flexibility allows the NFA to end up at the accepting state after the correct number of steps.
\item For the path containing the states $b_j$, at each state we can only transition to the next state in the path.
\end{enumerate}
The total number of states and transitions are both $O(dn)$.
\end{proof}

\begin{figure}[ht]
\centerline{\includegraphics[width=12cm]{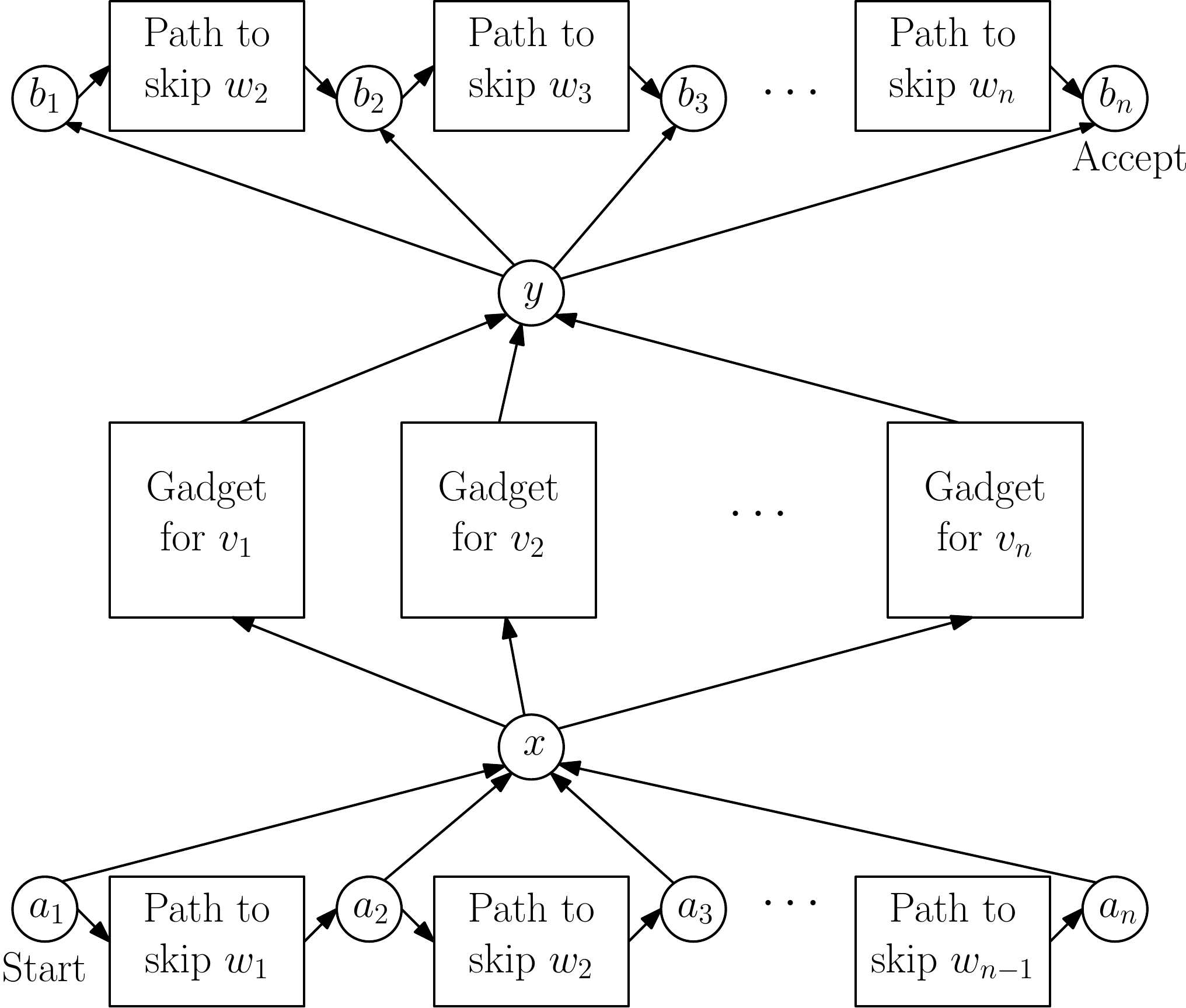}}
\caption{Structure of the acyclic NFA solving orthogonal vectors.}
\label{orthogonalitytester}
\end{figure}

\begin{corollary}
If there is an algorithm for {\tt NFA ACCEPTANCE} that runs 
in $O(n^{2-\epsilon})$ time, then {\rm SETH} is false.
\end{corollary}

\begin{proof}
Such an algorithm would imply an algorithm for {\tt ORTHOGONAL
VECTORS} that runs in the same time bound.  We then use
a result of Williams (e.g., \cite{Williams:2005} or
\cite[Thm.~1, p.~22]{Williams:2015}).
\end{proof}

\section{Conclusion}
In this paper, we analyzed the complexity of the acyclic NFA acceptance problem, which asks if an acyclic NFA accepts a given input. In the case where the NFA is unary, we showed that there is an algorithm using matrix multiplication that runs in $\tilde{O}(n^{\omega})$ time, which in fact enumerates all input lengths that are accepted. We also showed that we can reduce the triangle detection problem to unary acyclic NFA acceptance; improving on this algorithm would imply a breakthrough for triangle detection. In the general case, we show that significantly improving the trivial $O(nm)$ algorithm (where $n$ is the input length and $m$ is the number of edges in the NFA) when $m$ is $O(n)$ would imply that the strong exponential time hypothesis (SETH) is false.  

That said, there are a number of open questions remaining. First, what bounds can we show for acyclic NFA acceptance when the NFA is dense? In particular, can we prove a $\tilde{\Omega}(n^3)$ lower bound under some assumption? Second, can we reduce acyclic unary NFA acceptance and acyclic NFA acceptance to other problems?

\bibliographystyle{new}
\bibliography{abbrevs,pot}

\end{document}